\documentclass[10pt,journal,twocolumn,twoside]{IEEEtran}
\usepackage{cite}
\usepackage{amsmath,amssymb,amsfonts}
\usepackage{algorithmic}
\usepackage{graphicx}
\usepackage{textcomp}
\def\BibTeX{{\rm B\kern-.05em{\sc i\kern-.025em b}\kern-.08em
		T\kern-.1667em\lower.7ex\hbox{E}\kern-.125emX}}

\usepackage{amsmath,mathtools}
\usepackage{amsthm}
\usepackage{amssymb}
\usepackage{bm}
\usepackage{xspace}
\usepackage{xcolor}
\usepackage{graphicx}
\usepackage{url}
\usepackage{framed}
\usepackage{float}
\usepackage{rotating}
\usepackage{verbatim}
\usepackage{listings}
\usepackage{lscape}
\usepackage[ruled,vlined,linesnumbered]{algorithm2e}
\usepackage[nolist]{acronym}
\usepackage{pgfplots}
\pgfplotsset{compat=1.15}
\usepackage{pgf}
\usepackage{tikz}
\usetikzlibrary{arrows,shapes.misc,chains,scopes}
\usetikzlibrary{decorations.pathreplacing}
\usepackage{cleveref}
\usepackage{pgfplotstable}

\usepackage{colortbl}
\usetikzlibrary{matrix}
\usetikzlibrary{calc}
\usetikzlibrary{fit}

\newcommand{\argmax}[1]{\underset{#1}{\operatorname{arg}\,\operatorname{max}}\;}

\newcommand{\norm}[1]{\lVert#1\rVert}
\newcommand\numeq[1]%
{\stackrel{\scriptscriptstyle(\mkern-1.5mu#1\mkern-1.5mu)}{=}}
\newcommand{\executeiffilenewer}[3]{%
	\ifnum\pdfstrcmp{\pdffilemoddate{#1}}%
	{\pdffilemoddate{#2}}>0%
	{\immediate\write18{#3}}\fi%
}

\graphicspath{{figures/}}

\newcommand{\vech}{\boldsymbol{h}}

\newcommand{\vecs}{\boldsymbol{s}}

\newcommand{\vecx}{\boldsymbol{x}}
\newcommand{\vecy}{\boldsymbol{y}}
\newcommand{\vecz}{\boldsymbol{z}}

\newcommand{\sss}[1]{\scalebox{1.1}{$\scriptscriptstyle \mathsf{#1}$}}
\newcommand{\vecxd}{\bm{x}^{\sss{d}}}
\newcommand{\vecxp}{\bm{x}^{\sss{p}}}

\newcommand{\vecyp}{\bm{y}^{\sss{p}}}
\newcommand{\inner}[1]{\langle#1\rangle}

\newcommand{\bmm}{\begin{matrix}}
	\newcommand{\emm}{\end{matrix}}
\newcommand{\bpm}{\begin{pmatrix}}
	\newcommand{\epm}{\end{pmatrix}}

\newcommand{\bsbm}{\left[\begin{smallmatrix}}
	\newcommand{\esbm}{\end{smallmatrix}\right]}

\newcommand{\bbm}{\begin{bmatrix}}
	\newcommand{\ebm}{\end{bmatrix}}

\usepackage[utf8]{inputenc}
\usepackage[english]{babel}

\usepackage{booktabs}
\usepackage{multirow}
\usepackage{siunitx}
\usepackage{smartdiagram}

\theoremstyle{definition}
\newtheorem{theorem}{Theorem}

\newtheorem{corollary}{Corollary}
\newtheorem{lemma}{Lemma}
\newtheorem{remark}{Remark}

\begin{acronym}
	\acro{AWGN}{additive white Gaussian noise} 
	\acro{bpcu}{bits per channel use}
	\acro{ML}{maximum likelihood} 
	\acro{PM}{path metric} 
	\acro{BPSK}{binary phase-shift keying}
	\acro{CRC}{cyclic redundancy check}
	\acro{QPSK}{quadrature phase-shift keying}
	\acro{FER}{frame error rate}
	\acro{SC}{successive cancellation}
	\acro{SCL}{successive cancellation list}
	\acro{PAT}{pilot-assisted transmission}
	\acro{PCT}{polar-coded transmission}
	\acro{CSI}{channel state information}
	\acro{SNR}{signal-to-noise ratio}
	\acro{LLR}{log-likelihood ratio}
	\acro{LDPC}{low-density parity-check}
	\acro{MoM}{method of moments}
	\acro{QUP}{quasi-uniform puncturing}
	\acro{i.i.d.}{independent and identically distributed}
	\acro{r.v.}{random variable}
	\acro{SISO}{single-input single-output}
\end{acronym}

\newcommand{\own}{\textcolor{black}}
\newcommand{\ed}{\textcolor{black}}

\begin{document}
	\title{Polar-Coded Non-Coherent Communication}
	
	\author{\IEEEauthorblockN{Peihong Yuan, \IEEEmembership{Student Member, IEEE}, Mustafa Cemil Co\c{s}kun,                    \IEEEmembership{Student Member, IEEE},\\ Gerhard Kramer, \IEEEmembership{Fellow, IEEE}}
		\thanks{This work was supported by the German Research Foundation (DFG) under Grant KR~3517/9-1, and by the Helmholtz Gemeinschaft through the HGF-Allianz DLR@Uni project Munich Aerospace under the grant ``Efficient Coding and Modulation for Satellite Links with Severe Delay Constraints''.}
		\thanks{The authors are with the Institute for Communications Engineering of the Technical University of Munich (TUM), Theresienstr. 90, 80333 Munich, Germany (email: \{peihong.yuan,mustafa.coskun,gerhard.kramer\}@tum.de).}
	}
	\maketitle
	\begin{abstract}
		A polar-coded transmission (PCT) scheme with joint channel estimation and decoding is proposed for channels with unknown channel state information (CSI). The CSI is estimated via successive cancellation (SC) decoding and the constraints imposed by the frozen bits. SC list decoding with an outer code improves performance, \ed{including resolving a phase ambiguity when using quadrature phase-shift keying (QPSK) and Gray labeling.} Simulations with 5G polar codes and QPSK show gains of up to $2$~dB at a frame error rate (FER) of $10^{-4}$ over pilot-assisted transmission for various non-coherent models. Moreover, PCT performs within a few tenths of a dB to a coherent receiver with perfect CSI. For Rayleigh block-fading channels, PCT outperforms an FER upper bound based on random coding and within one dB of a lower bound.
	\end{abstract}
	\begin{IEEEkeywords}
		polar codes, fading channel, blind estimation, non-coherent communication, pilot-assisted transmission
	\end{IEEEkeywords}
	
	\section{Introduction}
	\IEEEPARstart{T}{he} communication setting where \ac{CSI} is not available at the transmitter or receiver is known as \emph{non-coherent} communication\ed{\cite[Ch.~10.7]{biglieri_wireless05}}. A common approach to address the lack of CSI is to embed pilot symbols in the transmitted symbol string, have the receiver estimate the \ac{CSI} based on the pilots, and use the estimated \ac{CSI} to decode. This approach is called \ac{PAT}\cite{Tong04:PAT} with mismatched decoding \cite[Ex.~5.22]{gallager1968information},\cite{Merhav94:MM,Lapidoth98:MM,TB05,TC07,Scarlett14:MM}.
	
	\ac{PAT} has two disadvantages for short block lengths: mismatched decoding reduces reliability and pilot symbols reduce rate significantly at low to moderate \ac{SNR} \cite{TB05,TC07,ODS19,Durisi16:Short,Liva17:Mismatched}. Both problems can be partially mitigated with sophisticated signal processing. For instance, one may use iterative channel estimation and decoding \cite{meyr1997digital,henkbook,herzet2007code,Noels03,Dauwels04,Herzet07:iterative,KhaBout:EM-APP06}, or two-stage algorithms that consider pilot symbols as part of the codebook\cite{CLO19,XCL19}, or even \ac{ML} decoding. Nevertheless, there is a fundamental performance degradation due to using pilot symbols\cite{ODS19}.
	
	We propose a pilot-free two-stage \ac{PCT} scheme to jointly estimate the \ac{CSI} and data with an adjustable complexity that can be made comparable to \ac{PAT}. In the first stage, \ac{SCL} decoding and the polar code constraints are used to estimate the \ac{CSI}. In the second stage, mismatched \ac{SCL} decoding proceeds with with this estimate. Gains of up to $2$ dB are shown at a \ac{FER} of $10^{-4}$ as compared to classic \ac{PAT} schemes for several non-coherent settings.
	
	A related method to estimate \ac{CSI} uses the parity-check constraints of a \ac{LDPC} code\cite{Imad2009,Gallager63:LDPC}. However, \ac{SCL} decoding of polar codes naturally provides soft estimates of frozen bits. Moreover, polar codes are usually used with a high-rate outer code~\cite{tal2015list,qualcomm_codes} that can resolve \ac{CSI} ambiguities, e.g., the phase ambiguity \ed{when using \ac{QPSK} and Gray labeling}\cite{Imad2009}. Of course, one may consider outer codes for \ac{LDPC} codes as well. Other low-complexity methods for non-coherent channels are described in, e.g.,\cite{WM02:noncoherent,CKM03:noncoherent,CT07,Imad2009,MLP13}. We remark that our focus is on \ac{QPSK} but the ideas extend to higher-order modulations. One may also combine \ac{PAT} and \ac{PCT} to optimize performance.
	
	This paper is organized as follows. Sec.~\ref{sec:preliminaries} introduces notation, the system model, polar codes, and \ac{PAT}. Sec.~\ref{sec:main} describes our joint channel estimation and decoding algorithm. Sec.~\ref{sec:numerical} demonstrates the effectiveness of the method for short polar codes concatenated with an outer \ac{CRC} code and \ac{QPSK}. Sec.~\ref{sec:conclusions} concludes the paper.

	\section{Preliminaries}
	\label{sec:preliminaries}
	Uppercase letters, e.g., $X$, denote random variables and lowercase letters, e.g., $x$, denote their realizations. The probability distribution of $X$ evaluated at $x$ is written as $P_X(x)$ or $P(x)$ when the argument is the lower-case version of the random variable. We similarly treat densities $p_X(x)$ or $p(x)$. For $a\le b$ we write $x_a^b$ for the row vector $(x_a,\dots,x_b)$. Lower case bold letters, e.g., $ \vecx $, also denote row vectors. Capital bold letters, e.g., $\bm{X}$, denote random vectors. All-zeros and all-ones vectors are denoted as $\boldsymbol{0}$ and $\boldsymbol{1}$, respectively. The notation $\overline{x_a^b}$ refers to the element-wise bit-flipped version of a binary vector $x_a^b$. We write $[N]=\left\{1,\dots,N\right\}$ and use calligraphic letters, e.g., $\mathcal{S}$, for sets otherwise. A subvector $x_\mathcal{S}$ of $x_1^N$ is formed by appropriately ordered elements with indices in $\mathcal{S}$. The cardinality of $\mathcal{S}$ is denoted as $|\mathcal{S}|$. We write $\norm{\cdot}$ for the $l_2$-norm and $\inner{\cdot,\cdot}$ for the inner product of two vectors.
	Finally, $\mathbb{F}^{\otimes m}$ refers to the $m$-fold Kronecker product of a matrix $\mathbb{F}$ where $\mathbb{F}^{\otimes 0}=1$.
	
	\subsection{System Model}
	
	Consider a scalar block-fading channel, i.e., the fading coefficient $H$ is constant for $n_c$ channel uses and changes independently across $B$ coherence blocks, resulting in a frame size of $n = B n_c$ symbols. The channel output of the $i$th coherence block is
	\begin{equation}\label{eq:model}
	\vecy_i = h_i\vecx_i+\vecz_i, \quad i = 1,\dots,B
	\end{equation}
	where $\vecx_i\in\mathcal{X}^{n_c}$ and $\vecy_i\in\mathbb{C}^{n_c}$  are the transmitted and received vectors, $h_i\in\mathbb{C}$ is a realization of $H$, and $\vecz_i$ is an \ac{AWGN} term whose entries are \ac{i.i.d.} as $\mathcal{CN}(0,2\sigma^2)$. Neither the transmitter nor the receiver knows $h_i$ or even the probability distribution of $H$. \ed{We assume that the noise variance $2\sigma^2$ is known to the receiver; this may be justified by the slow time scale of receiver device variations as compared to fading due to mobility.} A vector without subscripts denotes a concatenation of vectors or scalars, e.g., $\vecy=(\vecy_1,\dots,\vecy_B)$, $\vecx=(\vecx_1,\dots,\vecx_B)$ and $\vech=(h_1,\dots,h_B)$.
	
	Consider \ac{QPSK} with Gray labeling. \ed{The input alphabet is $\mathcal{X}=\left\{ \pm \Delta \pm j\Delta \right\}$, $\Delta>0$, and we map the binary vector $c_1^{2m}$ to $x_1^m\in\mathcal{X}^m$ via $\chi:\{0,1\}^{2m}\mapsto\mathcal{X}^{m}$ as 
		\begin{equation}\label{eq:label}
		\chi\left(c_1^{2m}\right) = \left(\chi_g(c_1,c_2),\chi_g(c_3,c_4),\dots,\chi_g(c_{2m-1},c_{2m})\right)
		\end{equation}
		where $\chi_g\left(c^2\right) = (-1)^{c_1}\Delta + j(-1)^{c_2}\Delta$. The mapping \eqref{eq:label} is \emph{symmetric}, i.e., if $\chi\left(c_1^{2m}\right) = \vecx$ then $\chi\left(\overline{c_1^{2m}}\right) = -\vecx$.}
	
	\subsection{Polar Codes}
	A binary polar code of block length \own{$N$ and dimension $K$ is defined by a set $\mathcal{A}\subseteq[N]$ of indices with $|\mathcal{A}| = K$ and the matrix $\mathbb{F}^{\otimes \log_2N}$, where $N$ is a positive-integer power of $2$ and $\mathbb{F}$ is the binary Hadamard matrix~\cite{arikan2009channel}}. Encoding is performed as $c_1^N=u_1^N\,\mathbb{F}^{\otimes\log_2N}$, where the \ed{input} vector $u_1^N$ has $K$ uniform information bits $u_\mathcal{A}$ and $N-K$ frozen bits $u_{\mathcal{F}}=\boldsymbol{0}$ with $\mathcal{F}=[N]\setminus\mathcal{A}$.
	A polar code is designed by storing the indices of the most reliable bits under SC decoding in the set $\mathcal{A}$~\cite{stolte2002rekursive,arikan2009channel}. In this work, we use the channel quality independent beta-expansion construction \cite{he2017beta}.
	
	An \ac{SC} decoder estimates the bit $u_i$ at decoding stage $i$ as $\hat{u}_i=0$ if
	$i\in\mathcal{F}$, and otherwise
	\begin{equation*}
	\hat{u}_i = \argmax {u_i\in\{0,1\}}
	p_{ \bm{Y},U_1^{i-1}|U_i}
	\left(\vecy,\hat{u}_1^{i-1}|u_i\right)
	\end{equation*}
	where the probabilities are approximated recursively by assuming that $U_j$, $i<j\leq \own{N}$, are \ac{i.i.d.} uniform random bits\cite{arikan2009channel}.
	Both encoding and SC decoding can be implemented with complexity $\mathcal{O}(N\log_2 N)$ \cite{arikan2009channel}.
	
	\ac{SCL} decoding with list size $L$ runs $L$ instances of an \ac{SC} decoder in parallel~\cite{tal2015list}. Each instance has a different hypothesis on the decoded information bits $\hat{u}_1^{i-1}$ at decoding stage $i$, called a decoding path. After decoding stage $N$, the decoder outputs the hypothesis of the most likely path as the estimate $\hat{u}_1^N$. An \ac{SCL} decoder can be implemented with complexity $\mathcal{O}(LN\log_2 N)$ \cite{tal2015list}. 
	
	Polar codes perform significantly better when combined with an outer \ac{CRC} code\cite{tal2015list}. Decoding proceeds as follows: An \ac{SCL} decoder for the inner polar code produces a list of codewords. The outer decoder discards those not fulfilling the constraints of the outer code. The decoder puts out the most likely of the remaining codewords if there is at least one, and it declares a frame error otherwise. For classic \ac{AWGN} channels, these modified polar codes are competitive under \ac{SCL} decoding for short block lengths\cite{Coskun18:Survey}. 
	
	\subsection{Pilot-Assisted Transmission}
	\label{sec:pragmatic}
	Consider \ac{PAT} as shown in Fig.~\ref{fig:bfc} where the first $n_p$ symbols in each coherence block are pilot symbols $\vecxp_i$ and the remaining \ed{$n_d=n_c-n_p$} symbols $\vecxd_i$ are coded. \ed{To keep the overall rate fixed, the $(N,K)$ code is punctured by using \ac{QUP} \cite{niu2013beyond} so that the code length after puncturing is $N_{\mathrm{punc}} = N - 2B n_p = 2B n_d$ with QPSK.} The pilot and coded symbols have the same energy. Upon observing $\vecy$, an \ac{ML} estimate of the \ac{CSI} is $\hat{h}_i = \inner{\vecyp_i,\vecxp_i}/\norm{\vecxp_i}^2$.
	A mismatched decoder uses $\hat{\vech}=(\hat{h}_1,\dots,\hat{h}_B)$ to compute the bit-wise \acp{LLR} that are fed to the \ac{SCL} decoder, leading to a codeword estimate.
	\begin{figure}    
		\centering
		\vspace{-2mm}
		\includegraphics[width=\columnwidth]{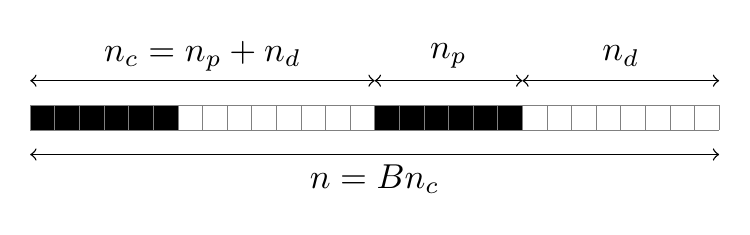}\\[-5.5mm]
		\caption{A \ac{PAT} frame structure with $B=2$ coherence blocks. Dark and white boxes represent pilot and coded symbols, respectively.}
		\label{fig:bfc}
	\end{figure}
	
	\section{Joint Channel Estimation and Decoding}
	\label{sec:main}
	This section presents a low-complexity joint channel estimation and decoding scheme for polar codes. We do not use pilot symbols, i.e., we have $n_p=0$ and $\vecx_i = \vecxd_i$. A random interleaver $\boldsymbol{\Pi}$ permutes the encoded bits $c_1^N$ and is followed by the mapping \eqref{eq:label}. The channel model is \eqref{eq:model}.
	
	Let $h_i=r_i e^{j\theta_i}$ where $r_i\in[0,\infty)$ and $\theta_i\in[0,2\pi)$, $i\in[B]$. We begin by estimating the amplitudes $r_i=|h_i|$ as
	\begin{equation}\label{eq:amplitude_est}
	\hat{r}_i=\left(\sqrt{2}\Delta\right)^{-1}\sqrt{\frac{1}{n_c}\norm{\vecy_i}^2-2\sigma^2}, \quad i = 1,\dots,B.
	\end{equation}
	
	Let $\beta$ be a number of input bits, and let $\mathcal{A}^{(\beta)}=\mathcal{\mathcal{A}}\cap[\beta]$ and $\mathcal{F}^{(\beta)}=\mathcal{F}\cap[\beta]$ be sets of information and frozen indices among the first $\beta$ input bits \ed{$u_1^\beta$}. We use the polar code constraints to estimate the phase as
	\begin{align}
	& \left\{\hat{\theta_1},\dots,\hat{\theta}_B\right\} = \argmax {\{\theta_1,\dots,\theta_B\}} p_{ \bm{Y}|U_{\mathcal{F}^{(\beta)}},\bm{H}}\left(\vecy\left|\boldsymbol{0},\hat{\bm{h}}\right.\right) \nonumber \\
	&\quad = \argmax {\{\theta_1,\dots,\theta_B\}} \sum_{u_{\mathcal{A}^{(\beta)}}} p_{ \bm{Y},U_{\mathcal{A}^{(\beta)}}|U_{\mathcal{F}^{(\beta)}},\bm{H}}\left(\vecy,u_{\mathcal{A}^{(\beta)}}\left|\boldsymbol{0},\hat{\bm{h}}\right.\right)
	\label{eq:proposed}
	\end{align}
	where $\hat{h}_i=\hat{r}_i e^{j\theta_i}$, $i\in[B]$. The sum in \eqref{eq:proposed} can be computed by \ac{SCL} decoding up to decoding stage $|\mathcal{F}^{(\beta)}|$ with a list size $L_e=2^{|\mathcal{A}^{(\beta)}|}$. To reduce complexity at the expense of accuracy, one can approximate the calculation with \ac{SCL} decoding and $L_e$ satisfying $1\le L_e < 2^{|\mathcal{A}^{(\beta)}|}$. In fact, simulations in Sec. \ref{sec:numerical} show that small list sizes such as $L_e=8$ give \ac{FER} curves close to those of the coherent receiver.
	\begin{remark}
		The search space in \eqref{eq:proposed} grows exponentially in the number of diversity branches $B$. There are several approaches to reduce complexity and we consider only the symmetry of the likelihood function due to the channel \eqref{eq:model} and mapping \eqref{eq:label} that halves the search space. We further adopt a coarse-fine search\cite{RB74,Imad2009} as an efficient optimizer. 
	\end{remark}
	
	\begin{lemma}\label{th:1}
		Polar-coded modulations with the mapping \eqref{eq:label} and the channel \eqref{eq:model} have a sign ambiguity for the channel coefficients, i.e., for all $\vecy$, $\vech$ and $u_1^{N-1}$, we have
		\begin{equation}
		p_{ \bm{Y}|U_1^{N},\bm{H}}\left( \vecy\left|(u_1^{N-1},0),\bm{h} \right.\right) = p_{\bm{Y}|U_1^{N},\bm{H}}\left( \vecy\left|(u_1^{N-1},1),-\bm{h} \right.\right).\nonumber
		\end{equation}
	\end{lemma}
	\begin{proof}
		For all $\vecx$, $\vecy$, $\vech$ and $\vecs\in\{-1,+1\}^B$, we have
		\begin{align}
		p\left(\vecy|\vecx,\vech\right) = \prod_{i=1}^{B} p_{\bm{Y}_i|\bm{X}_i,H_i}\left(\vecy_i|s_i\vecx_i,s_ih_i\right)\nonumber
		\end{align}
		as $s_i^2=1$. Recall that $c_1^N = \Pi^{-1}(\chi^{-1}(\vecx))$ so that $\overline{c_1^N} = \Pi^{-1}(\chi^{-1}(-\vecx))$. By choosing $\vecs = -\boldsymbol{1}$, we have
		\begin{align}\label{eq:proof2}
		p_{ \bm{Y}|\bm{C},\bm{H}}\left(\vecy|c_1^{N},\vech\right) = p_{\bm{Y}|\bm{C},\bm{H}}\left(\vecy|\overline{c_1^{N}},-\vech\right).
		\end{align}
		Let $u_1^N$ be the vector such that $c_1^N = u_1^{N}\mathbb{F}^{\otimes m}$. We have $\overline{c_1^N} = (u_1^{N-1},\overline{u_N})\mathbb{F}^{\otimes m}$ because the last row of $\mathbb{F}^{\otimes m}$ is $\boldsymbol{1}$.
	\end{proof}
	
	Lemma~\ref{th:1} implies that if a polar code is considered for \eqref{eq:model}, then the decoder cannot resolve the ambiguity on bit $u_N$. This ambiguity occurs for any binary linear block code \ed{that has a generator matrix with an all-ones row}, which is reflected in the bit $u_N$ for polar codes.
	
	\begin{theorem}\label{th:2}
		Polar-coded modulations with the mapping \eqref{eq:label} and the channel \eqref{eq:model} satisfy
		\begin{equation}\label{eq:theorem2}
		p\left( \vecy\left|{u_1^{i}},\vech \right.\right)
		=p_{ \bm{Y}|U_1^i,\bm{H}}\left( \vecy\left|{u_1^{i}},-\vech \right.\right)
		\end{equation}
		for all $\vecy$, $\vech$ and $u_1^{i}$, $i\in[N-1]$.
	\end{theorem}
	\begin{proof}
		For $i\in[N-1]$, we have
		\begin{align}
		p\left( \vecy|{u_1^{i}},\vech \right) &\numeq{\text{a}} \sum_{u_{i+1}^N} P\left( {u_{i+1}^{N}} \right) p\left( \vecy|{u_1^{N}},\vech \right)\nonumber\\
		&\numeq{\text{b}}\sum_{u_{i+1}^{N-1}} P\left( {u_{i+1}^{N-1}} \right)\left[\sum_{u_{N}}\frac{1}{2}p\left( \vecy|{u_1^{N}},\vech \right)\right] \nonumber\\
		&\numeq{\text{c}}\sum_{u_{i+1}^{N-1}} P\left( {u_{i+1}^{N-1}} \right)\left[\sum_{u_{N}}\frac{1}{2}p_{ \bm{Y}|U_1^N,\bm{H}}\left( \vecy|{u_1^{N}},-\vech \right)\right] \nonumber\\
		&\numeq{\text{d}} \sum_{u_{i+1}^N} P\left( {u_{i+1}^{N}} \right) p_{ \bm{Y}|U_1^N,\bm{H}}\left( \vecy|{u_1^{N}},-\vech \right) \nonumber
		\end{align}
		where step ($\text{a}$) follows by the law of total probability and the mutual independence of $U_1^i$, $U_{i+1}^N$ and $\bm{H}$; steps ($\text{b}$) and ($\text{d}$) follow by rearranging the sums and noting that $U_N$ is uniform; step ($\text{c}$) follows by Lemma~\ref{th:1}.
	\end{proof}
	\begin{corollary}\label{th:3}
		Polar-coded modulations with the mapping \eqref{eq:label} and the channel \eqref{eq:model} satisfy
		\begin{align}\label{eq:theorem3}
		p_{ \bm{Y}|U_{\mathcal{F}^{(\beta)}},\bm{H}}\left(\vecy\left|\boldsymbol{0},\vech\right.\right)=p_{ \bm{Y}|U_{\mathcal{F}^{(\beta)}},\bm{H}}\left(\vecy\left|\boldsymbol{0},-\vech\right.\right)
		\end{align}
		for all $\vecy$ and $\vech$.
	\end{corollary}
	\begin{proof}
		We expand
		\begin{align}
		p_{ \bm{Y}|U_{\mathcal{F}^{(\beta)}},\bm{H}}\left(\vecy\left|\boldsymbol{0},\vech\right.\right)
		&\numeq{\text{a}}\sum_{u_{\mathcal{A}^{(\beta)}}}P\left( {u_{\mathcal{A}^{(\beta)}}}\right)p\left( \vecy|{u_1^{\beta}},\vech \right)\nonumber\\
		&\numeq{\text{b}}\sum_{u_{\mathcal{A}^{(\beta)}}} P\left( {u_{\mathcal{A}^{(\beta)}}}\right) p_{ \bm{Y}|U_1^{\beta},\bm{H}}\left( \vecy|{u_1^{\beta}},-\vech \right)\nonumber
		\end{align}
		where step ($\text{a}$) follows by the law of total probability and mutually independent $U_{\mathcal{A}^{(\beta)}}$, $U_{\mathcal{F}^{(\beta)}}$ and $\bm{H}$; step ($\text{b}$) follows by Theorem~\ref{th:2}.
	\end{proof}
	
	Corollary~\ref{th:3} implies that the \ac{PCT} estimator outputs two solutions for~\eqref{eq:proposed}, namely $\{\hat{\theta}_1,\dots,\hat{\theta}_B\}$ and $\{\hat{\theta}_1+\pi,\dots,\hat{\theta}_B+\pi\}$ where addition is modulo $2\pi$. An outer code can resolve this ambiguity by  optimizing over the set $[0,2\pi)^{B-1}\times[0,\pi)$ to obtain $\{\hat{\theta}_1,\dots,\hat{\theta}_B\}$ by using the inner code constraints. The demodulator then feeds the \ac{SCL} decoder with the \acp{LLR}. Let $\mathcal{L}$ be the list of words $u_{\mathcal{A}}$ output by the decoder and define
	\begin{align}
	\mathcal{L}' = \{(u_{\mathcal{A}^{(N-1)}},\overline{u_N}): u_{\mathcal{A}}\in\mathcal{L}\}.\nonumber
	\end{align}
	The outer code now eliminates invalid words in $\mathcal{L}\cup\mathcal{L}'$. Among the survivors, if any, the estimate $\hat{u}_1^N$ is chosen to maximize $p_{ \bm{Y}|U_1^N,\bm{H}}( \vecy|{u_1^{N}},\hat{\vech} )$
	if $u_{\mathcal{A}}\in\mathcal{L}$ or
	$p_{ \bm{Y}|U_1^N,\bm{H}}( \vecy|{u_1^{N}},-\hat{\vech} )$
	if $u_{\mathcal{A}}\in\mathcal{L}'$. \own{An overview is given in Algorithm~1.}
	\begin{algorithm}[t]
		\caption{Blind Decoding Algorithm}
		\textbf{Input:} the received vector $y_1^n$. \\
		\textbf{Output:} the decoded word $\hat{u}_\mathcal{A}$.
		\vspace{-3mm}
		\begin{algorithmic}[1]
			\STATE estimate $\{\hat{r}_1,\dots,\hat{r}_B\}$ via \eqref{eq:amplitude_est}
			\STATE estimate $\{\hat{\theta}_1,\dots,\hat{\theta}_B\}\in[0,2\pi)^{B-1}\times[0,\pi)$ via \eqref{eq:proposed}
			\STATE run an \ac{SCL} decoder with the \acp{LLR} obtained using $\hat{\vech}$ and output the list $\mathcal{L}$ of $u_\mathcal{A}$
			\STATE obtain $\mathcal{L}'$ by flipping the last bit of all $u_\mathcal{A}\in\mathcal{L}$
			\STATE among all $u_\mathcal{A}\in\mathcal{L}\cup\mathcal{L}'$
			that pass the outer code test, choose the most likely one as $\hat{u}_\mathcal{A}$
		\end{algorithmic}
		\label{alg:1}
	\end{algorithm}
	
	\begin{remark}
		An outer code with a minimum distance of at least two can resolve the phase ambiguity.
	\end{remark}
	
	\section{Numerical Results}
	\label{sec:numerical}
	This section provides Monte Carlo simulation results to compare the performance of \ac{PAT} and \ac{PCT}. The \ac{SNR} is expressed as $E_s/N_0$, where $E_s$ is the energy per symbol and $N_0$ is the single-sided noise power spectral density. The inner code is a $(128,38)$ polar code and the outer code is a $6$-bit CRC code with generator polynomial $x^6+x^5+1$, resulting in a $(128,32)$ code. For the \ac{QPSK} modulator \eqref{eq:label} we have $n = B n_c = 64$ channel uses and an overall rate of $R=0.5$ \ac{bpcu}. \ed{For \ac{PAT}, the $(128,32)$ code is punctured to obtain $B n_p$ pilot bits in total, resulting in a $(128-2B n_p,32)$ code.} All curves shown in the figures below are for \ac{SCL} decoding with a list size of $L=8$ after estimating the \ac{CSI}. The optimization \eqref{eq:proposed} uses a coarse-fine search with $8$ levels in both the coarse and fine search parts\cite{RB74}. The performance is compared for various estimator parameters $\beta$ and $L_e$ \ed{and to the coherent receiver with perfect CSI. No puncturing is required for the coherent receiver.} As discussed below, the gains of our scheme are similar for $B\in\{1,2\}$ and with or without fading.
	
	\subsection{Single Coherence Block ($B=1$)}
	Consider the channel~\eqref{eq:model} with $B=1$, $r_1 = 1$, and uniformly distributed phase $\Theta_1\sim\mathcal{U}[0,2\pi)$. Fig.~\ref{fig:result1} compares \ac{PAT} and \ac{PCT}. The best \ac{PAT} performance for the \acp{FER} of interest \ed{was} achieved with $n_p = 14$, \ed{i.e., 14 pilot symbols gave the lowest \ac{SNR} for \acp{FER} ranging from $10^{-2}$ to $10^{-4}$ in Fig.~\ref{fig:result1}. For smaller $n_p$ the quality of the channel estimate limits performance, and for larger $n_p$ the puncturing weakens the polar code and limits performance.}
	
	\ac{PCT} performs within $0.3$ dB of the receiver with perfect \ac{CSI} if the estimator is run with $L_e = 8$ and up to the last frozen bit with $\beta=113$. It thereby outperforms \ac{PAT} by about $1.5~\text{dB}$ at a \ac{FER} of $10^{-4}$. Observe that if the estimator is run up to the last frozen bit before the first information bit, i.e., $\beta=47$, then the performance is worse than for \ac{PAT}. The parameters $\beta=113$ and $L_e=1$ provide a good trade-off between complexity and performance when combined with a second-stage SCL decoding with a list size $L=8$.
	\begin{figure}
		\vspace{-4mm}
		\includegraphics[width=\columnwidth]{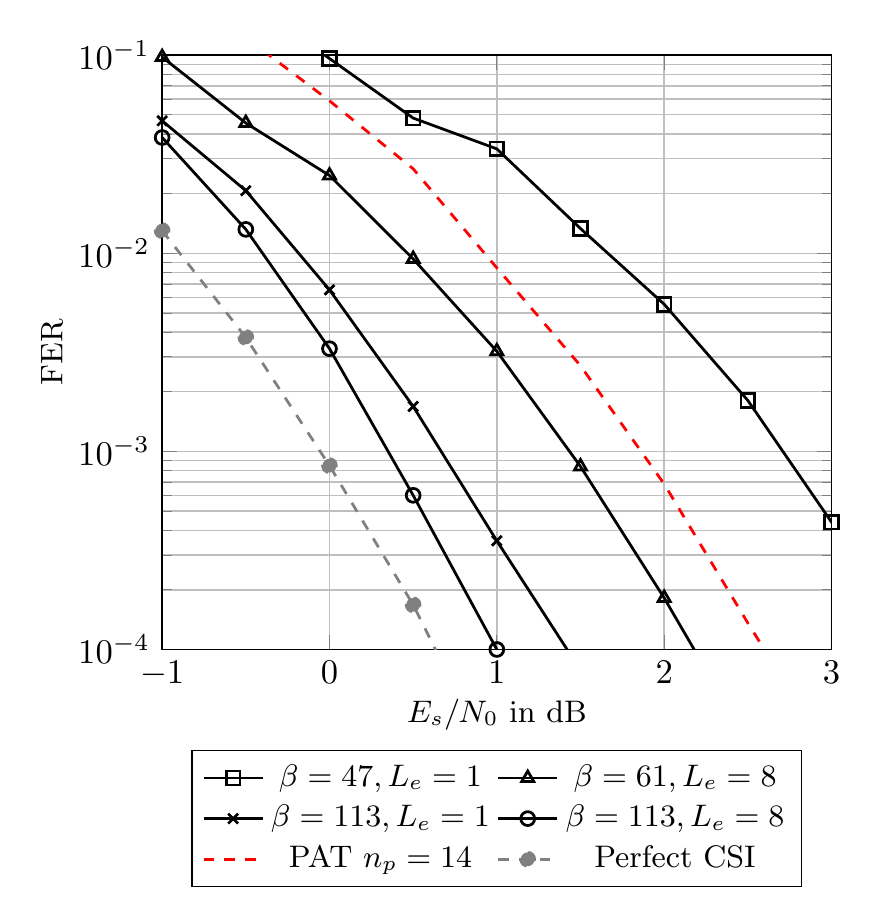}\\[-5.5mm]
		\caption{Performance of \ac{PAT} and \ac{PCT} for the channel~\eqref{eq:model} \own{with $B=1$, $r_1=1$, and $\Theta_1\sim[0,2\pi)$. A $(128,32)$ polar code was used with \ac{QPSK} so that $n = n_c =64$ and the overall rate is $R=0.5$ \ac{bpcu}}. SCL decoding uses a list size of $L=8$ for all cases.}
		\label{fig:result1}
	\end{figure}
	
	\ed{Table \ref{tab:complexity} compares the number of visited nodes per frame in the polar decoding tree along with the \ac{FER} at $E_s/N_0 = 1$~dB. Each visited node corresponds to an input bit (including the frozen bits) visited by the algorithm\cite[Remark 4]{SC-Fano19}. For PCT, we state the sum of the number of nodes visited by the estimator and the number of nodes visited by the decoder.} The number of visited nodes with \ac{PAT} and perfect \ac{CSI} is thus the same. Observe that \ac{PCT} with $\beta=113$ and $L_e=1$ visits a similar number of nodes as \ac{PAT} with a list size $L=32$ \ed{(the difference is less than $10\%$)} and it reduces the error probability by one order of magnitude.  We remark that measuring the complexity by the number of visited nodes is pessimistic for \ac{PCT} since most of the visited nodes are frozen bits. Hence, simplified \ac{SC} decoders\cite{YK11,SGV14} can significantly reduce complexity.
	\begin{table}
		\caption{Number of Visited Nodes per Frame at $E_s/N_0=1$ {d}B}
		\vspace{-2mm}
		\begin{tabular}{ccc}
			\hline\hline
			Method & FER & Visited Nodes \\
			\hline
			\ac{PAT} ($n_p=14$, $L=8$) & $8.43\times10^{-3}$ & $631$ \\
			\ac{PAT} ($n_p=14$, $L=32$) & $3.16\times10^{-3}$ & $2223$ \\
			\ac{PCT} ($\beta=47$, $L_e=1$, $L=8$) & $3.36\times10^{-2}$ & $1383$ \\
			\ac{PCT} ($\beta=61$, $L_e=8$, $L=8$) & $3.20\times10^{-3}$ & $2151$ \\
			\ac{PCT} ($\beta=113$, $L_e=1$, $L=8$) & $3.50\times10^{-4}$ & $2439$ \\
			\ac{PCT} ($\beta=113$, $L_e=8$, $L=8$) & $1.00\times10^{-4}$ & $8807$ \\
			Perfect \ac{CSI} ($L=8$) & $2.40\times10^{-5}$ & $631$ \\
			\hline\hline
		\end{tabular}
		\label{tab:complexity}
	\end{table}

	\subsection{Two Coherence Blocks ($B=2$)}
	We next consider $B=2$ coherence blocks. Fig.~\ref{fig:result2phase} shows the \ac{FER} for $r_i = 1$ and $\Theta_i\sim\mathcal{U}[0,2\pi)$, $i\in\{1,2\}$. Fig.~\ref{fig:result2rayleigh} shows the \ac{FER} for a Rayleigh block-fading channel with $H_i\sim\mathcal{CN}(0,1)$, $i\in\{1,2\}$. The best performance for \ac{PAT} \ed{was} achieved with $n_p = 7$ pilot symbols per coherence block for both cases. Observe that, in both cases, \ac{PCT} outperforms \ac{PAT} by about $2$ dB at a \ac{FER} $\approx10^{-4}$. Moreover, \ac{PCT} approaches the performance of a coherent receiver with perfect \ac{CSI}.
	\begin{figure}
		\vspace{-4mm}
		\includegraphics[width=\columnwidth]{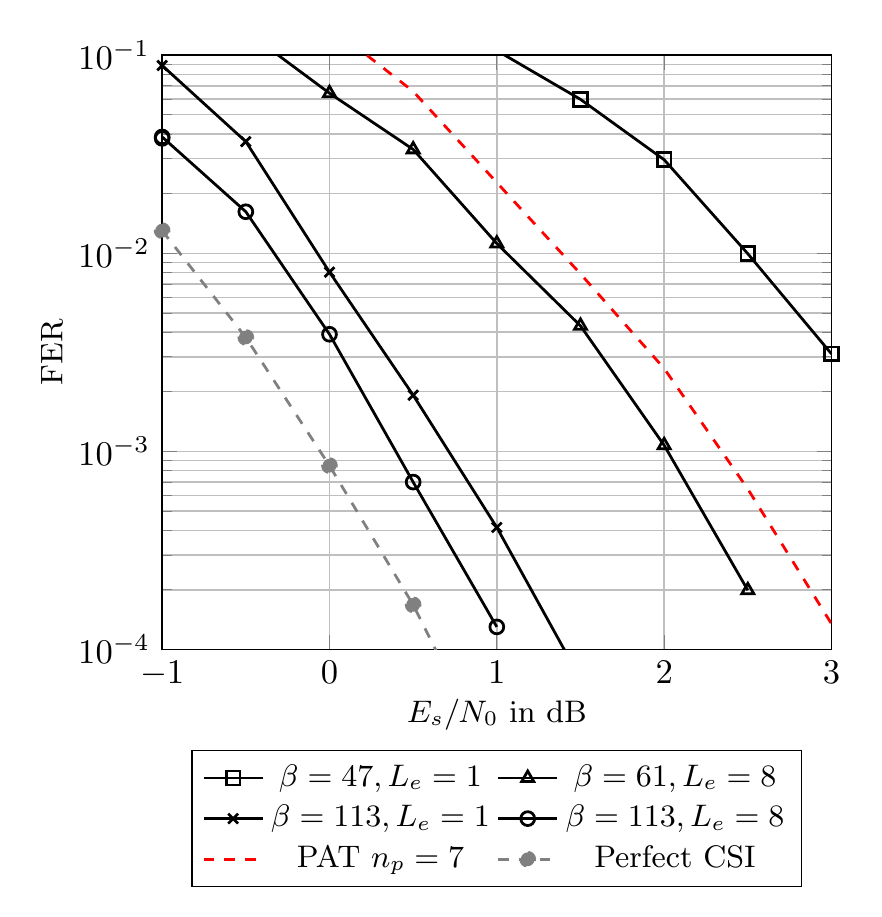}\\[-5.5mm]
		\caption{Performance of \ac{PAT} and \ac{PCT} for the channel~\eqref{eq:model} \own{with $B=2$, $r_i=1$, and $\Theta_i\sim[0,2\pi)$ for $i\in\{1,2\}$. A $(128,32)$ polar code was used with \ac{QPSK} so that $n = 2n_c = 64$ and the overall rate is $R=0.5$ \ac{bpcu}}. SCL decoding uses a list size of $L=8$ for all cases.}
		\label{fig:result2phase}
	\end{figure}
	
	Fig.~\ref{fig:result2rayleigh} also provides an upper (achievability) bound based on the random coding union bound with $s$ parameter (RCUs) \cite[Thm. 1]{martinez11-02a} and a lower (converse) bound called a metaconverse (MC) \cite[Thm. 28]{polyanskiy10-05a}. Both bounds assume that there is a power constraint per coherence block rather than a codeword. Also, the input distribution is induced by unitary space-time modulation. \own{For more details, see \cite{lancho20-saddle}.}
	\begin{figure}
		\vspace{-4mm}
		\includegraphics[width=\columnwidth]{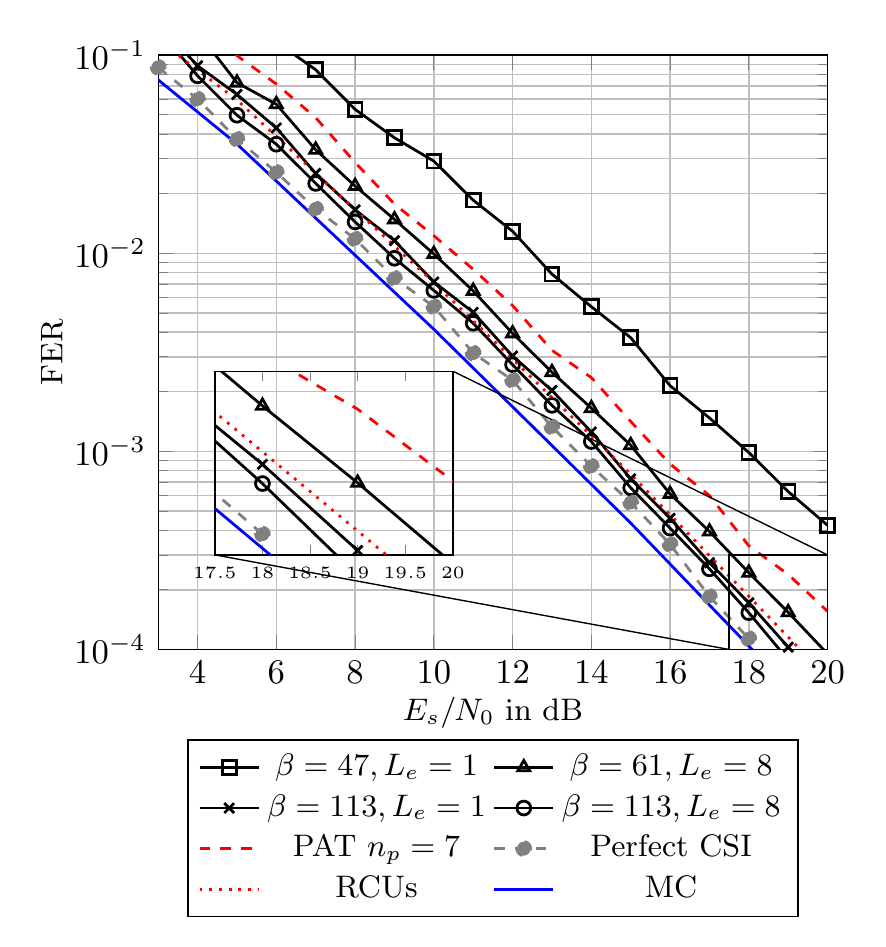}\\[-5.5mm]
		\caption{Performance of \ac{PAT} and \ac{PCT} for a Rayleigh block-fading channel \own{and $B=2$. A $(128,32)$ polar code was used with \ac{QPSK} and the overall rate is $R=0.5$ \ac{bpcu}}. SCL decoding uses a list size of $L=8$ for all cases.}
		\label{fig:result2rayleigh}
	\end{figure}
	
	\section{Conclusions}
	\label{sec:conclusions}
	A \ac{PCT} scheme was proposed that estimates \ac{CSI} via SCL decoding and the constraints imposed by the frozen bits. An outer code improves reliability and resolves phase ambiguities. Simulation results show that \ac{PCT} significantly outperforms \ac{PAT} schemes with a similar complexity and approaches the performance of a coherent receiver.
	
	\section*{Acknowledgements}
	The authors wish to thank Dr. A. Lancho (Chalmers) for providing the RCUs and MC bounds.
	
	

\begin{thebibliography}{10}
		\providecommand{\url}[1]{#1}
		\csname url@samestyle\endcsname
		\providecommand{\newblock}{\relax}
		\providecommand{\bibinfo}[2]{#2}
		\providecommand{\BIBentrySTDinterwordspacing}{\spaceskip=0pt\relax}
		\providecommand{\BIBentryALTinterwordstretchfactor}{4}
		\providecommand{\BIBentryALTinterwordspacing}{\spaceskip=\fontdimen2\font plus
			\BIBentryALTinterwordstretchfactor\fontdimen3\font minus
			\fontdimen4\font\relax}
		\providecommand{\BIBforeignlanguage}[2]{{%
				\expandafter\ifx\csname l@#1\endcsname\relax
				\typeout{** WARNING: IEEEtran.bst: No hyphenation pattern has been}%
				\typeout{** loaded for the language `#1'. Using the pattern for}%
				\typeout{** the default language instead.}%
				\else
				\language=\csname l@#1\endcsname
				\fi
				#2}}
		\providecommand{\BIBdecl}{\relax}
		\BIBdecl
		
		\bibitem{biglieri_wireless05}
		E.~Biglieri, \emph{Coding for Wireless Channels}.\hskip 1em plus 0.5em minus
		0.4em\relax Springer, 2005.
		
		\bibitem{Tong04:PAT}
		L.~Tong, B.~M. Sadler, and M.~Dong, ``Pilot-assisted wireless transmissions:
		General model, design criteria, and signal processing,'' \emph{{IEEE} Signal
			Process. Mag.}, vol.~21, no.~6, pp. 12--25, Nov. 2004.
		
		\bibitem{gallager1968information}
		R.~G. Gallager, \emph{Information Theory and Reliable Communication}.\hskip 1em
		plus 0.5em minus 0.4em\relax John Wiley \& Sons, Inc., 1968.
		
		\bibitem{Merhav94:MM}
		N.~{Merhav}, G.~{Kaplan}, A.~{Lapidoth}, and S.~{Shamai Shitz}, ``On
		information rates for mismatched decoders,'' \emph{IEEE Trans. Inf. Theory},
		vol.~40, no.~6, pp. 1953--1967, 1994.
		
		\bibitem{Lapidoth98:MM}
		A.~{Lapidoth} and P.~{Narayan}, ``Reliable communication under channel
		uncertainty,'' \emph{IEEE Trans. Inf. Theory}, vol.~44, no.~6, pp.
		2148--2177, 1998.
		
		\bibitem{TB05}
		G.~{Taricco} and E.~{Biglieri}, ``Space-time decoding with imperfect channel
		estimation,'' \emph{IEEE Trans. Wireless Commun.}, vol.~4, no.~4, pp.
		1874--1888, 2005.
		
		\bibitem{TC07}
		G.~{Taricco} and G.~{Coluccia}, ``Optimum receiver design for correlated
		{Rician} fading {MIMO} channels with pilot-aided detection,'' \emph{IEEE J.
			Sel. Areas Commun.}, vol.~25, no.~7, pp. 1311--1321, 2007.
		
		\bibitem{Scarlett14:MM}
		J.~{Scarlett}, A.~{Martinez}, and A.~G. i.~{Fabregas}, ``Mismatched decoding:
		Error exponents, second-order rates and saddlepoint approximations,''
		\emph{IEEE Trans. Inf. Theory}, vol.~60, no.~5, pp. 2647--2666, 2014.
		
		\bibitem{ODS19}
		J.~\"{O}stman, G.~{Durisi}, E.~G. Str\"om, M.~C. Co\c{s}kun, and G.~{Liva},
		``Short packets over block-memoryless fading channels: Pilot-assisted or
		noncoherent transmission?'' \emph{IEEE Trans. Commun.}, vol.~67, no.~2, pp.
		1521--1536, Feb. 2019.
		
		\bibitem{Durisi16:Short}
		G.~Durisi, T.~Koch, and P.~Popovski, ``Towards massive, ultra-reliable, and
		low-latency wireless communications with short packets,'' \emph{Proc.
			{IEEE}}, vol. 104, no.~9, pp. 1711--1726, Sep. 2016.
		
		\bibitem{Liva17:Mismatched}
		G.~Liva, G.~Durisi, M.~Chiani, S.~S. Ullah, and S.~C. Liew, ``{Short codes with
			mismatched channel state information: A case study},'' in \emph{IEEE Int.
			Workshop on Signal Process. Adv. in Wireless Commun.}, Sapporo, Japan, Jul.
		2017, pp. 1--5.
		
		\bibitem{meyr1997digital}
		H.~Meyr, M.~Moeneclaey, and S.~Fechtel, \emph{Digital Communication Receivers:
			Synchronization, Channel Estimation, and Signal Processing}.\hskip 1em plus
		0.5em minus 0.4em\relax Wiley, 1997.
		
		\bibitem{henkbook}
		H.~Wymeersch, \emph{Iterative Receiver Design}.\hskip 1em plus 0.5em minus
		0.4em\relax Cambridge, 2007.
		
		\bibitem{herzet2007code}
		C.~Herzet, N.~Noels, V.~Lottici, H.~Wymeersch, M.~Luise, M.~Moeneclaey, and
		L.~Vandendorpe, ``Code-aided turbo synchronization,'' \emph{Proc. IEEE},
		vol.~95, no.~6, pp. 1255--1271, 2007.
		
		\bibitem{Noels03}
		N.~{Noels}, C.~{Herzet}, A.~{Dejonghe}, V.~{Lottici}, H.~{Steendam},
		M.~{Moeneclaey}, M.~{Luise}, and L.~{Vandendorpe}, ``Turbo synchronization:
		an {EM} algorithm interpretation,'' in \emph{IEEE Int. Conf. Commun.},
		vol.~4, 2003, pp. 2933--2937 vol.4.
		
		\bibitem{Dauwels04}
		J.~{Dauwels} and H.~A. {Loeliger}, ``Phase estimation by message passing,'' in
		\emph{IEEE Int. Conf. Commun.}, vol.~1, 2004, pp. 523--527 Vol.1.
		
		\bibitem{Herzet07:iterative}
		C.~{Herzet}, V.~{Ramon}, and L.~{Vandendorpe}, ``A theoretical framework for
		iterative synchronization based on the sum–product and the
		expectation-maximization algorithms,'' \emph{IEEE Trans. Signal Process.},
		vol.~55, no.~5, pp. 1644--1658, 2007.
		
		\bibitem{KhaBout:EM-APP06}
		M.~Khalighi and J.~J. Boutros, ``Semi-blind channel estimation using the {EM}
		algorithm in iterative {MIMO} {APP} detectors,'' \emph{{IEEE} Trans. Wireless
			Commun.}, vol.~5, no.~11, pp. 3165--3173, Nov. 2006.
		
		\bibitem{CLO19}
		M.~C. Co\c{s}kun, G.~{Liva}, J.~\"{O}stman, and G.~{Durisi}, ``Low-complexity
		joint channel estimation and list decoding of short codes,'' in \emph{ITG
			Int. Conf. Syst., Commun. and Coding}, Feb 2019.
		
		\bibitem{XCL19}
		M.~{Xhemrishi}, M.~C. {Coşkun}, G.~{Liva}, J.~{Östman}, and G.~{Durisi},
		``List decoding of short codes for communication over unknown fading
		channels,'' in \emph{Asilomar Conf. Signals, Systems, Computers}, 2019, pp.
		810--814.
		
		\bibitem{Imad2009}
		R.~Imad, S.~Houcke, and M.~Ghogho, ``\BIBforeignlanguage{English}{Blind
			estimation of the phase and carrier frequency offsets for {LDPC}-coded
			systems},'' \emph{\BIBforeignlanguage{English}{EURASIP J. Adv. Signal
				Process.}}, vol. 2010, no.~1, pp. 1--13, 2010.
		
		\bibitem{Gallager63:LDPC}
		R.~G. Gallager, \emph{Low-density parity-check codes}.\hskip 1em plus 0.5em
		minus 0.4em\relax Cambridge, MA, USA: M.I.T. Press, 1963.
		
		\bibitem{tal2015list}
		I.~Tal and A.~Vardy, ``List decoding of polar codes,'' \emph{{IEEE} Trans. Inf.
			Theory}, vol.~61, no.~5, pp. 2213--2226, May 2015.
		
		\bibitem{qualcomm_codes}
		``{LS} on channel coding,'' 3GPP TSG RAN WG1 Meeting, R1-1715317, Prague, Czech
		Republic, Tech. Rep.~90, Aug. 2017.
		
		\bibitem{WM02:noncoherent}
		D.~{Warrier} and U.~{Madhow}, ``Spectrally efficient noncoherent
		communication,'' \emph{IEEE Trans. Inf. Theory}, vol.~48, no.~3, pp.
		651--668, 2002.
		
		\bibitem{CKM03:noncoherent}
		{Rong-Rong Chen}, R.~{Koetter}, U.~{Madhow}, and D.~{Agrawal}, ``Joint
		noncoherent demodulation and decoding for the block fading channel: a
		practical framework for approaching {Shannon} capacity,'' \emph{IEEE Trans.
			Commun.}, vol.~51, no.~10, pp. 1676--1689, 2003.
		
		\bibitem{CT07}
		G.~{Coluccia} and G.~{Taricco}, ``An optimum blind receiver for correlated
		{Rician} fading {MIMO} channels,'' \emph{IEEE Commun. Lett.}, vol.~11, no.~9,
		pp. 738--739, 2007.
		
		\bibitem{MLP13}
		B.~{Matuz}, G.~{Liva}, E.~{Paolini}, M.~{Chiani}, and G.~{Bauch}, ``Low-rate
		non-binary {LDPC} codes for coherent and blockwise non-coherent {AWGN}
		channels,'' \emph{IEEE Trans. Commun.}, vol.~61, no.~10, pp. 4096--4107,
		2013.
		
		\bibitem{arikan2009channel}
		E.~Ar{\i}kan, ``Channel polarization: A method for constructing
		capacity-achieving codes for symmetric binary-input memoryless channels,''
		\emph{{IEEE} Trans. Inf. Theory}, vol.~55, no.~7, pp. 3051--3073, Jul. 2009.
		
		\bibitem{stolte2002rekursive}
		N.~Stolte, ``{Rekursive Codes mit der Plotkin-Konstruktion und ihre
			Decodierung},'' Ph.D. dissertation, TU Darmstadt, 2002.
		
		\bibitem{he2017beta}
		G.~He, J.-C. Belfiore, I.~Land, G.~Yang, X.~Liu, Y.~Chen, R.~Li, J.~Wang,
		Y.~Ge, R.~Zhang \emph{et~al.}, ``Beta-expansion: A theoretical framework for
		fast and recursive construction of polar codes,'' in \emph{IEEE Global
			Commun. Conf.}, 2017, pp. 1--6.
		
		\bibitem{Coskun18:Survey}
		M.~C. Co\c{s}kun, G.~Durisi, T.~Jerkovits, G.~Liva, W.~Ryan, B.~Stein, and
		F.~Steiner, ``Efficient error-correcting codes in the short blocklength
		regime,'' \emph{Elsevier Phys. Commun.}, vol.~34, pp. 66--79, Jun. 2019.
		
		\bibitem{niu2013beyond}
		K.~Niu, K.~Chen, and J.-R. Lin, ``Beyond turbo codes: Rate-compatible punctured
		polar codes,'' \emph{IEEE Int. Conf. Commun.}, pp. 3423--3427, Jun. 2013.
		
		\bibitem{RB74}
		D.~{Rife} and R.~{Boorstyn}, ``Single tone parameter estimation from
		discrete-time observations,'' \emph{IEEE Trans. Inf. Theory}, vol.~20, no.~5,
		pp. 591--598, 1974.
		
		\bibitem{SC-Fano19}
		M.~{Jeong} and S.~{Hong}, ``{SC-Fano} decoding of polar codes,'' \emph{IEEE
			Access}, vol.~7, pp. 81\,682--81\,690, 2019.
		
		\bibitem{YK11}
		A.~{Alamdar-Yazdi} and F.~R. {Kschischang}, ``A simplified
		successive-cancellation decoder for polar codes,'' \emph{IEEE Commun. Lett.},
		vol.~15, no.~12, pp. 1378--1380, 2011.
		
		\bibitem{SGV14}
		G.~{Sarkis}, P.~{Giard}, A.~{Vardy}, C.~{Thibeault}, and W.~J. {Gross}, ``Fast
		polar decoders: Algorithm and implementation,'' \emph{IEEE J. Sel. Areas
			Commun.}, vol.~32, no.~5, pp. 946--957, 2014.
		
		\bibitem{martinez11-02a}
		A.~Martinez and A.~{Guill{\'e}n i F{\`a}bregas}, ``Saddlepoint approximation of
		random--coding bounds,'' in \emph{Inf. Theory Applic. Workshop (ITA)}, San
		Diego, CA, U.S.A., Feb. 2011.
		
		\bibitem{polyanskiy10-05a}
		Y.~Polyanskiy, H.~V. Poor, and S.~Verd\'u, ``Channel coding rate in the finite
		blocklength regime,'' \emph{{IEEE} Trans. Inf. Theory}, vol.~56, no.~5, pp.
		2307--2359, May 2010.
		
		\bibitem{lancho20-saddle}
		A.~{Lancho}, J.~{Östman}, G.~{Durisi}, T.~{Koch}, and G.~{Vazquez-Vilar},
		``Saddlepoint approximations for short-packet wireless communications,''
		\emph{IEEE Trans. Wireless Commun.}, vol.~19, no.~7, pp. 4831--4846, 2020.
		
	\end{thebibliography}
\end{document}